\documentclass[a4paper]{article}
\usepackage[utf8]{inputenc}
\usepackage{amsmath}
\usepackage{amsfonts}
\usepackage{amssymb}
\usepackage{amsthm}
\usepackage{bbm}
\usepackage{xcolor}
\usepackage{graphicx}
\usepackage{hyperref}
\numberwithin{equation}{section}
\usepackage{geometry}
\usepackage{caption}
\usepackage{subcaption}
\DeclareMathOperator\supp{supp}

\newtheorem{theorem}{Theorem}[section]
\newtheorem{lemma}[theorem]{Lemma}
\newtheorem{definition}[theorem]{Definition}
\newtheorem{proposition}[theorem]{Proposition}
\newtheorem{corollary}[theorem]{Corollary}

\newtheorem{remark}[theorem]{Remark}
\newtheorem{example}[theorem]{Example}

\theoremstyle{definition}
\newtheorem{assumption}{Assumption}
\newtheorem*{notation*}{Notation}

\title{Power law in Sandwiched Volterra Volatility model}

\author{
    G. Di Nunno$^{1,2}$\\ \href{mailto:giulian@math.uio.no}{giulian@math.uio.no}
    \and
    A. Yurchenko-Tytarenko$^{1,}$\footnote{Corresponding author} \\ \href{mailto:antony@math.uio.no}{antony@math.uio.no}
}

\date{%
    $^1$Department of Mathematics, University of Oslo\\
    $^2$Department of Business and Management Science, NHH Norwegian School of Economics, Bergen\\
    [2ex]%
    \today
}

\begin{document}

\maketitle

\begin{abstract}
     In this paper, we present analytical proof demonstrating that the Sandwiched Volterra Volatility (SVV) model is able to reproduce the power-law behavior of the at-the-money implied volatility skew, provided the correct choice of the Volterra kernel. To obtain this result, we assess the second-order Malliavin differentiability of the volatility process and investigate the conditions that lead to explosive behavior in the Malliavin derivative. As a supplementary result, we also prove a general Malliavin product rule.
\end{abstract}

\textbf{Acknowledgements.} The present research is carried out within the frame and support of the ToppForsk project nr. 274410 of the Research Council of Norway with the title STORM: Stochastics for Time-Space Risk Models.

\section{Introduction}

One of the well-established benchmarks for evaluating option pricing models is comparing the model-generated Black-Scholes implied volatility surface $(\tau, \kappa) \mapsto \widehat \sigma (\tau, \kappa)$ with the empirically observed one $(\tau, \kappa) \mapsto \widehat \sigma_{\text{emp}} (\tau, \kappa)$. In this context, $\tau$ represents the \textit{time to maturity} and $\kappa := \log\frac{K}{e^{r\tau} S_0}$ is the \textit{log-moneyness} with $K$ denoting the strike, $S_0$ the current price of an underlying asset and $r$ being the instantaneous interest rate. In particular, for any fixed $\tau$, the values of $\widehat \sigma_{\text{emp}}(\tau,\kappa)$ plotted against $\kappa$ are known to produce convex ``\textit{smiley}'' patterns with negative slopes at-the-money (i.e. when $\kappa\approx 0$). Furthermore, as reported in e.g. \cite{Delemotte_Marco_Segonne_2023, Fouque_Papanicolaou_Sircar_Solna_2004, Gatheral_Jaisson_Rosenbaum_2014} or \cite[Subsection 2.2]{Di_Nunno_Kubilius_Mishura_Yurchenko-Tytarenko_2023}, the smile at-the-money becomes progressively steeper as $\tau \to 0$ with a rule-of-thumb behavior
\begin{equation}\label{eq: empirical power law intro}
    \left|\frac{\widehat \sigma_{\text{emp}} (\tau, \kappa) - \widehat \sigma_{\text{emp}} (\tau, \kappa')}{\kappa - \kappa'}\right| \propto \tau^{-\frac{1}{2} + H}, \quad \kappa,\kappa' \approx 0, \quad H \in \left(0, \frac 1 2\right).
\end{equation}
The phenomenon \eqref{eq: empirical power law intro} is known as the \textit{power law} of the at-the-money implied volatility skew, and if one wants to replicate it, one may look for a model with
\begin{equation}\label{eq: theoretical power law intro}
    \left|\frac{\partial \widehat \sigma}{\partial \kappa} (\tau, \kappa)\right|_{\kappa = 0} = O(\tau^{-\frac{1}{2} + H}), \quad \tau \to 0.
\end{equation}
However, it turns out that the property \eqref{eq: theoretical power law intro} is not easy to obtain: for example, as discussed in \cite[Section 7.1]{Alos_Leon_Vives_2007} or \cite[Remark 11.3.21]{Lee_2006}, classical Brownian diffusion stochastic volatility models fail to produce implied volatilities with power law \eqref{eq: theoretical power law intro}. In the literature, \eqref{eq: theoretical power law intro} is usually replicated by introducing a volatility process with very low H\"older regularity within the \textit{rough} volatility framework popularized by Gatheral, Jaisson and Rosenbaum in their landmark paper \cite{Gatheral_Jaisson_Rosenbaum_2014}. The efficiency of this approach can be explained as follows. 
\begin{itemize}
    \item On the one hand, a theoretical result of Fukasawa \cite{Fukasawa_2021} suggests that the volatility process cannot be H\"older continuous of a high order in continuous non-arbitrage models exhibiting the property \eqref{eq: theoretical power law intro}. In other words, the roughness of volatility is, in some sense, a necessary condition to reproduce \eqref{eq: theoretical power law intro} (at least in the fully continuous setting).

    \item On the other hand, as proved in the seminal 2007 paper \cite{Alos_Leon_Vives_2007} of Alòs, León and Vives, the short-term explosion \eqref{eq: theoretical power law intro} of the implied volatility skew can be deduced from the explosion of the Malliavin derivative of volatility. In particular, the latter characteristic is exhibited by \textit{fractional Brownian motion with $H<1/2$}, a common driver in rough volatility literature.
\end{itemize}

However, despite the ability to reproduce the power law \eqref{eq: theoretical power law intro}, rough volatility models are not perfect. In particular,
\begin{itemize}
    \item[--] in the specific context of fractional Brownian motion, roughness contradicts the observations \cite{Bollerslev_Mikkelsen_1996, Ding_Granger_1996, Ding_Granger_Engle_1993, Lobato_Velasco_2000, Willinger_Taqqu_Teverovsky_1999} of long memory on the market;

    \item[--] in addition, volatility processes with long memory seem to be better in replicating the shape of implied volatility for longer maturities \cite{Comte_Renault_1998, Funahashi_Kijima_2017, Funahashi_Kijima_2017-1};

    \item[--] furthermore, there is no guaranteed procedure of transition between physical and pricing measures: it is not always clear whether the volatility process $\sigma = \{\sigma(t),~t\in[0,T]\}$ hits zero and therefore the integral $\int_0^t \frac{1}{\sigma^2(s)}ds$ that is typically present in martingale densities (see e.g. \cite{BGP2000}) may be poorly defined; 

    \item[--] just like many classical Brownian stochastic volatility models (see e.g. \cite{Andersen_Piterbarg_2006}), they may suffer from moment explosions in price, which results in complications with the pricing of some assets, quadratic hedging, and numerical methods.
\end{itemize}
For more details on rough volatility, we refer the reader to a recent review \cite[Subsection 3.3.2]{Di_Nunno_Kubilius_Mishura_Yurchenko-Tytarenko_2023} or a regularly updated literature list on the subject \cite{Rough_volatility_literature}.

Recently, a series of papers \cite{DiNunno_Mishura_Yurchenko-Tytarenko_2022, DiNunno_Mishura_Yurchenko-Tytarenko_2023b, DNMYT_2023_spdbhn} introduced the \textit{Sandwiched Volterra Volatility (SVV) model} which accounts for all the problems mentioned above. More precisely, the volatility process $Y=\{Y(t),~t\in[0,T]\}$ is assumed to follow the stochastic differential equation
\[
    Y(t) = y_0 + \int_0^t b(s, Y(s))ds + Z(t)
\]
driven by a general H\"older continuous Gaussian Volterra process $Z(t) = \int_0^t \mathcal K(t,s)dB(s)$. The special part of the equation above is the drift $b$. It is assumed that there are two continuous functions $0 < \varphi < \psi$ such that for some $\varepsilon > 0$
\begin{align*}
    b(t,y) &\ge \frac{C}{(y-\varphi(t))^\gamma}, & y \in(\varphi(t), \varphi(t) + \varepsilon),
    \\
    b(t,y) &\le -\frac{C}{(\psi(t) - y)^\gamma}, & y \in(\psi(t) - \varepsilon, \psi(t)).
\end{align*}
Such an explosive nature of the drift resembling the one in SDEs for Bessel processes (see e.g. \cite[Chapter XI]{Revuz_Yor_1999}) or singular SDEs of \cite{Hu2008} ensures that, with probability 1,
\[
    0 < \varphi(t) < Y(t) < \psi(t),
\]
which immediately solves the moment explosion problem (see e.g. \cite[Theorem 2.6]{DiNunno_Mishura_Yurchenko-Tytarenko_2022}) and allows for a transparent transition between physical and pricing measures \cite[Subsection 2.2]{DiNunno_Mishura_Yurchenko-Tytarenko_2022}. In addition, the flexibility in the choice of the kernel $\mathcal K$ \textit{should} allow to replicate both the long memory and the power law behavior \eqref{eq: theoretical power law intro}.

The main goal of this paper is to give the theoretical justification to the latter claim: we prove that, with the correct choice of the Volterra kernel $\mathcal K$, the SVV model indeed reproduces \eqref{eq: theoretical power law intro}. In order to do that, we employ the fundamental result \cite[Theorem 6.3]{Alos_Leon_Vives_2007} of Alòs, León and Vives mentioned above and check that the Malliavin derivative $DY(t)$ indeed exhibits explosive behavior. The difficulty of this approach is as follows. While the first-order Malliavin differentiability of $Y(t)$ is established in \cite[Section 3]{DiNunno_Mishura_Yurchenko-Tytarenko_2022} with
\begin{equation*}
    D_s Y(t) = \mathcal K(t,s) + \int_s^t \mathcal K(u,s) b'_y(u, Y(u)) \exp\left\{\int_u^t b'_y(v,Y(v))dv\right\}du,
\end{equation*}
\cite[Theorem 6.3]{Alos_Leon_Vives_2007} actually demands the existence of the second-order Malliavin derivative. In principle, it is intuitively clear how this derivative should look like:
\begin{equation}\label{eq: wrong computation}
\begin{aligned}
    D_rD_sY(t) &=  D_r \int_s^t \mathcal K(u,s) b'_y(u, Y(u)) \exp\left\{\int_u^t b'_y(v,Y(v))dv\right\}du
    \\
    & =  \int_s^t \mathcal K(u,s) D_r\left[b'_y(u, Y(u)) \exp\left\{\int_u^t b'_y(v,Y(v))dv\right\}\right]du
    \\
    & = \int_s^t \mathcal K(u,s) \exp\left\{\int_u^t b'_y(v,Y(v))dv\right\} D_r\left[b'_y(u, Y(u))\right] du
    \\
    &\quad + \int_s^t \mathcal K(u,s) b'_y(u, Y(u)) D_r\left[\exp\left\{\int_u^t b'_y(v,Y(v))dv\right\}\right]du
    \\
    & = \int_s^t \mathcal K(u,s) b''_{yy}(u, Y(u))\exp\left\{\int_u^t b'_y(v,Y(v))dv\right\} D_r\left[Y(u)\right] du
    \\
    &\quad + \int_s^t \mathcal K(u,s) b'_y(u, Y(u)) \exp\left\{\int_u^t b'_y(v,Y(v))dv\right\} \int_u^t b''_{yy}(v,Y(v)) D_r[Y(v)]dvdu.
\end{aligned}
\end{equation}
However, the computations in \eqref{eq: wrong computation} are far from straightforward to be justified. For example, the functions $y\mapsto b'_y(t,y)$ and $y\mapsto b''_{yy}(t,y)$ demonstrate explosive behavior as $y\to\varphi(t)+$ and $y\to\psi(t)-$ for any $t\in[0,T]$. This makes it impossible to use the classical Malliavin chain rules such as \cite[Proposition 1.2.3]{Nualart_2006} requiring boundedness of the derivative or \cite[Proposition 1.2.4]{Nualart_2006} demanding the Lipschitz condition. In order to overcome this issue, we have to use some special properties of the volatility process established in \cite{DNMYT_2023_spdbhn} and tailor a version of the Malliavin chain rule specifically for our needs.

The paper is organized as follows. In Section \ref{sec: sandwich facts}, we provide some necessary details about the sandwiched volatility process $Y$. In Section \ref{sec: Malliavin}, we prove second-order Malliavin differentiability of $Y(t)$. Finally, in Section \ref{sec: power law}, we use \cite[Theorem 6.3]{Alos_Leon_Vives_2007} to determine conditions on the kernel under which the SVV model reproduces \eqref{eq: theoretical power law intro}. In Appendix \ref{appendix}, we gather some necessary facts from Malliavin calculus, list some of the notation and, in addition, we prove a general Malliavin product rule to fit our purposes and that we were not able to find in the literature. 

\section{Preliminaries on sandwiched processes}\label{sec: sandwich facts}

In this section, we gather all the necessary details about the main object of our study: the class of \emph{sandwiched processes driven by H\"older-continuous Gaussian Volterra noises}.

Fix some $T\in(0,\infty)$ and consider a kernel $\mathcal K: [0,T]^2 \to \mathbb R$ satisfying the following assumptions.

\begin{assumption}\label{assum: K}
    The kernel $\mathcal K$ is of Volterra type, i.e. $\mathcal K(t,s) = 0$ whenever $t\le s$, and
    \begin{itemize}
        \item[(K1)] $\mathcal K$ is square-integrable, i.e.
        \[
            \int_0^T\int_0^T \mathcal K^2(t,s)dsdt < \infty,
        \]

        \item[(K2)] there exists $H\in(0,1)$ such that for all $\lambda\in(0,H)$ and $0 \le t_1 \le t_2 \le T$
        \[
            \int_0^{T} (\mathcal K(t_2,s) - \mathcal K(t_1,s))^2ds \le C_\lambda |t_2-t_1|^\lambda,
        \]
        where $C_\lambda > 0$ is some constant depending on $\lambda$.
    \end{itemize}
\end{assumption}

\begin{remark}\label{rem: K}
    Note that items (K1) and (K2) of Assumption \ref{assum: K} jointly imply that
    \begin{equation}\label{eq: sup of K}
        \sup_{t\in[0,T]}\int_0^T \mathcal K^2(t,s)ds < \infty.
    \end{equation}
\end{remark}
Let $B = \{B(t),~t\in[0,T]\}$ be a standard Brownian motion. Assumption \ref{assum: K} allows to define a \textit{Gaussian Volterra process}
\begin{equation}\label{eq: def Z}
    Z(t) := \int_0^t \mathcal K(t,s)dB(s), \quad t\in [0,T],
\end{equation}
and, moreover, Assumption \ref{assum: K}(K2) together with \cite[Theorem 1 and Corollary 4]{ASVY2014} imply that $Z$ has a modification with H\"older continuous trajectories of any order $\lambda\in(0,H)$. In what follows, we always use this modification of $Z$: in other words, with probability 1, for any $\lambda\in(0,H)$ there exists a random variable $\Lambda = \Lambda(\lambda)>0$ such that for all $0\le t_1 \le t_2 \le T$
\begin{equation}\label{eq: Z Holder}
    |Z(t_2) - Z(t_1)| \le \Lambda |t_2-t_1|^\lambda.
\end{equation}
Furthermore, as stated in \cite[Theorem 1]{ASVY2014}, the random variable $\Lambda$ from \eqref{eq: Z Holder} can be chosen such that
\begin{equation}\label{eq: Lambda moments}
    \mathbb E[\Lambda^r] < \infty \quad \text{for all }r\in\mathbb R.
\end{equation}
In what follows, we assume that \eqref{eq: Lambda moments} always holds.

Next, denote
\begin{equation}\label{eq: D}
\begin{aligned}
    \mathcal D &:= \{(t,y)\in[0,T]\times\mathbb R~|~\varphi(t) < y < \psi(t)\}, 
    \\
    \overline{\mathcal D} &:= \{(t,y)\in[0,T]\times\mathbb R~|~\varphi(t) \le y \le \psi(t)\}.
\end{aligned}
\end{equation}
Take $H\in(0,1)$ from Assumption \ref{assum: K}(K2), consider two $H$-H\"older continuous functions $\varphi$, $\psi$: $[0,T] \to \mathbb R$ such that
\[
    0 < \varphi(t) < \psi(t) \quad \text{for all }t\in[0,T],
\]
and define a function $b$: $\mathcal D \to \mathbb R$ as
\begin{equation}\label{eq: b}
    b(t,y) := \frac{\theta_1(t)}{(y - \varphi(t))^{\gamma_1}} - \frac{\theta_2(t)}{(\psi(t) - y)^{\gamma_2}} + a(t,y),
\end{equation}
where the coefficients in \eqref{eq: b} satisfy the following assumption.

\begin{assumption}\label{assum: b}
    The constants $\gamma_1$, $\gamma_2> 0$ and functions $\theta_1$, $\theta_2$, $a$ are such that
    \begin{itemize}
        \item[(B1)] $\gamma_1 > \frac{1}{H} - 1$, $\gamma_2 > \frac{1}{H} - 1$ with $H\in(0,1)$ being from Assumption \ref{assum: K}(K2);

        \item[(B2)] the functions $\theta_1$, $\theta_2$: $[0,T]\to\mathbb R$ are strictly positive and continuous;

        \item[(B3)] the function $a$: $[0,T]\times\mathbb R\to\mathbb R$ is locally Lipschitz in $y$ uniformly in $t$, i.e. for any $N > 0$ there exists a constant $C_N>0$ that does not depend on $t$ such that
        \[
            |a(t,y_2) - a(t, y_1)| \le C_N|y_2 - y_1|, \quad t\in [0,T], \quad y_1, y_2 \in [-N,N];
        \]
        
        \item[(B4)] $a$: $[0,T]\times\mathbb R\to\mathbb R$ is two times differentiable w.r.t. the spatial variable $y$ with $a$, $a'_y$, $a''_{yy}$ all being continuous on $[0,T]\times\mathbb R$.
    \end{itemize}
\end{assumption}

\begin{remark}\label{rem: exp bounded}
    Note that $b'_y$ is bounded from above on $\mathcal D$: indeed,
    \begin{align*}
        b'_y(t,y) & = -\frac{\gamma_1\theta_1(t)}{(y - \varphi(t))^{\gamma_1+1}} - \frac{\gamma_2\theta_2(t)}{(\psi(t) - y)^{\gamma_2+1}} + a'_y(t,y)
        \\
        & < \max_{(t,y)\in\overline{\mathcal D}} a'_y(t,y) < \infty.
    \end{align*}
\end{remark}
Finally, fix $\varphi(0) < y_0 < \psi(0)$ and consider a stochastic differential equation of the form
\begin{equation}\label{eq: sandwiched SDE}
    Y(t) = y_0 + \int_0^t b(s, Y(s))ds + Z(t), \quad t\in [0,T].
\end{equation}
By \cite[Theorem 4.1]{DNMYT_2023_spdbhn}, under Assumptions \ref{assum: K} and \ref{assum: b}, the SDE \eqref{eq: sandwiched SDE} has a unique strong solution $Y=\{Y(t),~t\in[0,T]\}$. Moreover, with probability 1,
\begin{equation}\label{eq: sandwich property}
    \varphi(t) < Y(t) < \psi(t) \quad \text{for all } t\in [0,T].
\end{equation}
\begin{remark}
    Motivated by the property \eqref{eq: sandwich property}, we will call the solution $Y$ of \eqref{eq: sandwiched SDE} a \textit{sandwiched} process.
\end{remark}

In what follows, we will need to analyze the behavior of stochastic processes $|b(t,Y(t))|$, $|b'_y(t,Y(t))|$ and $|b''(t,Y(t))|$, $t\in[0,T]$. In this regard, the property \eqref{eq: sandwich property} alone is not sufficient: the process $Y$ can, in principle, approach the bounds $\varphi$ and $\psi$ which results in an explosive growth of the processes mentioned above. Luckily, \cite[Theorem 4.2]{DNMYT_2023_spdbhn} provides a refinement of \eqref{eq: sandwich property} allowing for a more precise control of $Y$ near $\varphi$ and $\psi$. We give a slightly reformulated version of this result below.

\begin{theorem}\label{th: main property of Y}
    Let Assumptions \ref{assum: K} and \ref{assum: b} hold and $\lambda\in(0,H)$, $\Lambda = \Lambda(\lambda) > 0$ be from \eqref{eq: Z Holder}. Then there exist deterministic constants $C_Y = C_Y(\lambda)>0$ and $\beta = \beta(\lambda)>0$ such that
    \begin{equation*}
        \varphi(t) + \frac{C_Y}{(1+\Lambda)^\beta} \le Y(t) \le \psi(t) - \frac{C_Y}{(1+\Lambda)^\beta} \quad \text{for all }t\in[0,T].
    \end{equation*}
    In particular, since $\Lambda$ can be chosen to have moments of all orders, for all $r\ge 0$
    \begin{equation*}
        \mathbb E\left[\sup_{t\in[0,T]}\frac{1}{(Y(t) - \varphi(t))^r}\right] < \infty, \quad  \mathbb E\left[\sup_{t\in[0,T]}\frac{1}{(\psi(t) - Y(t))^r}\right] < \infty.
    \end{equation*}
\end{theorem}

We finalize this section by citing the first-order Malliavin differentiability result for the sandwiched process \eqref{eq: sandwiched SDE} proved in \cite[Section 3]{DiNunno_Mishura_Yurchenko-Tytarenko_2022}.

\begin{theorem}\label{th: 1 order Malliavin differentiability}
    Let Assumptions \ref{assum: K} and \ref{assum: b} hold and $Y$ be the sandwiched process given by \eqref{eq: sandwiched SDE}. Then, for any $t\in[0,T]$, $Y(t) \in \mathbb D^{1,2}$ and, with probability 1, for a.a. $s\in [0,T]$
    \begin{equation}\label{eq: Malliavin derivative of Y(t)}
        D_s Y(t) = \mathcal K(t,s) + \int_s^t \mathcal K(u,s) b'_y(u, Y(u)) \exp\left\{\int_u^t b'_y(v,Y(v))dv\right\}du.
    \end{equation}
\end{theorem}

\begin{remark}
    The result above actually holds for more general drifts than the one given in \eqref{eq: b}. The same is also, in principle, true for the results of the subsequent sections. Namely, it would be sufficient to assume that there exist deterministic constants $c>0$, $r>0$, $\gamma > \frac{1}{H} - 1$ and $0< y_* < \max_{t\in[0,T]}|\psi(t) - \varphi(t)|$ such that
    \begin{itemize}
        \item $b$: $\mathbb D \to \mathbb R$ is continuous on $\mathbb D$ and has continuous partial derivatives $b'_y$, $b''_{yy}$;

        \item for any $0 < \varepsilon < \frac{1}{2}\max_{t\in[0,T]}|\psi(t) - \varphi(t)|$,
        \[
            |b(t,y_2) - b(t,y_1)| \le \frac{c}{\varepsilon^r}|y_2 - y_1|, \quad t\in[0,T], \quad \varphi(t)+\varepsilon \le y_1 \le y_2 \le \psi(t) - \varepsilon;
        \]
        
        \item $b$ has an explosive growth to $\infty$ near $\varphi$ and explosive decay to $-\infty$ near $\psi$ of order $\gamma > \frac{1}{H} - 1$, i.e.
        \begin{align*}
            b(t,y) &\ge \frac{c}{(y-\varphi(t))^\gamma}, & y \in(\varphi(t), \varphi(t) + y_*),
            \\
            b(t,y) &\le -\frac{c}{(\psi(t) - y)^\gamma}, & y \in(\psi(t) - y_*, \psi(t));
        \end{align*}

        \item for all $(t,y) \in \mathcal D$, the partial derivatives $b'_y$ and $b''_{yy}$ satisfy
        \[
            - C\left( 1 + \frac{c}{(y-\varphi(t))^r} + \frac{c}{(\psi(t) - y)^r}\right) < b'_y(t,y) < C
        \]
        and
        \[
            |b''_{yy}| \le C\left( 1 + \frac{c}{(y-\varphi(t))^r} + \frac{c}{(\psi(t) - y)^r}\right).
        \]
    \end{itemize}
    However, since \eqref{eq: b} is the most natural choice satisfying these assumptions, we stick to this shape for notational convenience.
\end{remark}

\section{Second-order Malliavin differentiability}\label{sec: Malliavin}

Let Assumptions \ref{assum: K} and \ref{assum: b} hold and $Y = \{Y(t),~t\in[0,T]\}$ be the sandwiched process defined by \eqref{eq: sandwiched SDE} with the drift \eqref{eq: b}.

\begin{notation*}
    Here and in the sequel, $C$ will denote any positive deterministic constant the exact value of which is not relevant. Note that $C$ may change from line to line (or even within one line).
\end{notation*}

The main goal of this section is to establish second-order Malliavin differentiability of the sandwiched process \eqref{eq: sandwiched SDE} and compute the corresponding derivative explicitly. As mentioned above, the main difficulty lies in controlling the behavior of $b(t,Y(t))$, $b'_y(t,Y(t))$ and $b''_{yy}(t,Y(t))$ whenever $Y(t)$ approaces the bounds. Luckily, Theorem \ref{th: main property of Y} gives all the necessary tools to do that as summarized in the following proposition.
\begin{proposition}\label{prop: bounds for derivatives}
    There exists a random variable $\xi > 0$ such that 
    \begin{itemize}
        \item for any $p\ge 1$, $\mathbb E[\xi^p] < \infty$;

        \item for any $t\in[0,T]$,
        \[
            |b(t,Y(t))| + |b'_y(t,Y(t))| + |b''_{yy}(t,Y(t))| < \xi.
        \]
    \end{itemize}
    In particular, for any $p\ge 1$,
    \[
        \mathbb E\left[ \sup_{t\in[0,T]} \left(|b(t,Y(t))|^p + |b'_y(t,Y(t))|^p + |b''_{yy}(t,Y(t))|^p\right) \right] < \infty.
    \]
\end{proposition}

\begin{proof}
    Fix $\lambda\in(0,H)$ and take the corresponding $\Lambda > 0$ from \eqref{eq: Z Holder} and $C_Y, \beta > 0$ be from Theorem \ref{th: main property of Y}. Then
    \begin{align*}
        |b(t,Y(t))| &= \frac{|\theta_1(t)|}{(Y(t) - \varphi(t))^{\gamma_1}} + \frac{|\theta_2(t)|}{(\psi(t) - Y(t))^{\gamma_2}} + |a(t, Y(t))|
        \\
        & \le \frac{\sup_{t\in[0,T]}|\theta_1(t)|(1+\Lambda)^{\beta{\gamma_1}}}{C_Y^{\gamma_1}} + \frac{\sup_{t\in[0,T]}|\theta_2(t)|(1+\Lambda)^{\beta\gamma_2}}{C_Y^{\gamma_2}} + \sup_{(t,y) \in \mathcal D}|a(t,y)|
        \\
        &:= \xi_0,
    \end{align*}
    \begin{align*}
        |b'_y(t,Y(t))| &= \frac{\gamma_1|\theta_1(t)|}{(Y(t) - \varphi(t))^{\gamma_1+1}} + \frac{\gamma_2|\theta_2(t)|}{(\psi(t) - Y(t))^{\gamma_2+1}} + |a'_y(t, Y(t))|
        \\
        & \le \frac{\gamma_1\sup_{t\in[0,T]}|\theta_1(t)|(1+\Lambda)^{\beta(\gamma_1+1)}}{C_Y^{\gamma_1+1}} 
        \\
        &\quad + \frac{\gamma_2\sup_{t\in[0,T]}|\theta_2(t)|(1+\Lambda)^{\beta(\gamma_2 + 1)}}{C_Y^{\gamma_2+1}} 
        \\
        &\quad + \sup_{(t,y) \in \mathcal D}|a'_y(t,y)|
        \\
        &:= \xi_1,
    \end{align*}
    \begin{align*}
        |b''_{yy}(t,Y(t))| &= \frac{\gamma_1(\gamma_1+1)|\theta_1(t)|}{(Y(t) - \varphi(t))^{\gamma_1+2}} + \frac{\gamma_2(\gamma_2+1)|\theta_2(t)|}{(\psi(t) - Y(t))^{\gamma_2+2}} + |a''_{yy}(t, Y(t))|
        \\
        & \le \frac{\gamma_1(\gamma_1+1)\sup_{t\in[0,T]}|\theta_1(t)|(1+\Lambda)^{\beta(\gamma_1+2)}}{C_Y^{\gamma_1+2}} 
        \\
        &\quad+ \frac{\gamma_2(\gamma_2+1)\sup_{t\in[0,T]}|\theta_2(t)|(1+\Lambda)^{\beta(\gamma_2 + 2)}}{C_Y^{\gamma_2+2}} 
        \\
        &\quad+ \sup_{(t,y) \in \mathcal D}|a''_{yy}(t,y)|
        \\
        &:= \xi_2.
    \end{align*}
    Note that $\xi_0$, $\xi_1$ and $\xi_2$ have moments of all orders by the properties of $\Lambda$, see \eqref{eq: Lambda moments}, and hence, putting
    \[
        \xi := \xi_0 + \xi_1 + \xi_2,
    \]
    we obtain the required result.
\end{proof}

As noted in Theorem \ref{th: 1 order Malliavin differentiability}, $Y(t)\in\mathbb D^{1,2}$ for each $t\ge 0$. In fact, Proposition \ref{prop: bounds for derivatives} together with the shape \eqref{eq: Malliavin derivative of Y(t)} of the derivative allow to establish a more general result.

\begin{proposition}\label{prop: D1p}
    For any $t\in [0,T]$ and $p > 1$, $Y(t) \in \mathbb D^{1,p}$.
\end{proposition}

\begin{proof}
    Note that, by \eqref{eq: sandwich property}, $\mathbb E[|Y(t)|^p]<\infty$ for any $p>1$, so, by Lemma \ref{lemma: moments and Malliavin differentiability} from the Appendix, it is sufficient to prove that
    \[
        \mathbb E\left[\left( \int_0^T \left(D_s Y(t)\right)^2 ds \right)^{\frac p 2}\right] < \infty
    \]
    for any $p>1$. Note that, by Remark \ref{rem: exp bounded},
    \[
        \exp\left\{\int_s^t b'_y(v,Y(v))dv\right\} < \exp\left\{cT\right\},
    \]
    where 
    \[
        c := \max_{(t,y)\in\overline{\mathcal D}} a'_y(t,y)
    \]
    and, by Proposition \ref{prop: bounds for derivatives}, there exists a random variable $\xi$ having all moments such that
    \[
        \sup_{s\in[0,T]}|b'_y(s, Y(s))| \le \xi.
    \]\
    Hence
    \begin{equation}\label{proofeq: pathwise bound on the derivative}
    \begin{aligned}
        |D_s Y(t)| & \le |\mathcal K(t,s)| + \int_s^t |\mathcal K(u,s)| |b'_y(u, Y(u))| \exp\left\{\int_u^t b'_y(v,Y(v))dv\right\}du
        \\
        & \le |\mathcal K(t,s)| + \xi \exp\left\{cT\right\} \int_s^t |\mathcal K(u,s)| du.
    \end{aligned}    
    \end{equation}
    By Assumption \ref{assum: K} and Remark \ref{rem: K},
    \[
        \left(\int_0^T \mathcal K^2(t,s)ds\right)^\frac{p}{2} < \infty,
    \]
    therefore
    \begin{equation}\label{proofeq: bound for norm of D1}
    \begin{aligned}
        \mathbb E&\left[\left( \int_0^T \left(D_s Y(t)\right)^2 ds \right)^{\frac p 2}\right]
        \\
        & \le C \left(\int_0^T \mathcal K^2(t,s)ds\right)^\frac{p}{2}
        \\
        &\quad +C\mathbb E \left[\left( \int_0^T \int_0^t \mathcal K^2(u,s) (b'_y(u, Y(u)))^2 \exp\left\{2\int_u^t b'_y(v,Y(v))dv\right\}du ds \right)^{\frac p 2}\right]
        \\
        &\le C \left(\int_0^T \mathcal K^2(t,s)ds\right)^\frac{p}{2} + C \mathbb E \left[\xi^p\right]\exp\left\{pcT\right\} \left( \int_0^T \int_0^t \mathcal K^2(u,s) du ds \right)^{\frac p 2}
        \\
        &< \infty,
    \end{aligned}    
    \end{equation}
    
    which ends the proof.
\end{proof}

Our next goal is to establish the Malliavin chain rule for the random variables $b'_y (t, Y(t))$ and $\exp\left\{ \int_u^t b'_y(v,Y(v))dv \right\}$.

\begin{proposition}\label{prop: auxiliary derivative computations}
    For any $0 \le u \le t \le T$ and $p>1$, 
    \begin{itemize}
        \item[1)] $b'_y (t, Y(t)) \in \mathbb D^{1,p}$ with
        \begin{equation}\label{eq: Malliavin derivative of b'}
            D_s \left[ b'_y (t, Y(t)) \right] = b''_{yy} (t, Y(t)) D_sY(t),
        \end{equation}
        
        \item[2)] $\exp\left\{ \int_u^t b'_y(v,Y(v))dv \right\} \in \mathbb D^{1,p}$
        with
        \begin{equation}\label{eq: Malliavin derivative of exp(b')}
            D_s \left[\exp\left\{ \int_u^t b'_y(v,Y(v))dv \right\}\right] = \exp\left\{ \int_u^t b'_y(v,Y(v))dv \right\} \int_u^t b''_{yy}(v,Y(v)) D_s Y(v) dv.
        \end{equation}
    \end{itemize} 
\end{proposition}

\begin{proof}
    1) We shall start from proving that $b'_y (t, Y(t)) \in \mathbb D^{1,p}$. Note that $b'_y$ is not a bounded function itself and it does not have bounded derivatives -- hence the classical chain rule from \cite[Section 1.2]{Nualart_2006} cannot be applied here in a straightforward manner. In order to overcome this issue, we will use the approach in the spirit of \cite[Lemma A.1]{OcK1991} or \cite[Proposition 3.4]{DiNunno_Mishura_Yurchenko-Tytarenko_2022}. For the reader's convenience, we divide the proof into steps.

    \textbf{Step 0.} First of all, observe that $b'(t,Y(t)) \in L^2(\Omega)$ as a direct consequence of Proposition \ref{prop: bounds for derivatives}. Also, for any $p>1$,
    \begin{equation*}
    \begin{aligned}
        \mathbb E\left[\left( \int_0^T \left(b''_{yy}(t,Y(t)) D_s Y(t)\right)^2 ds \right)^{\frac p 2}\right] < \infty.
    \end{aligned}    
    \end{equation*}
    Indeed, again by Proposition \ref{prop: bounds for derivatives} together with the proof of Proposition \ref{prop: D1p}, we have
    \begin{align*}
        \mathbb E\left[\left( \int_0^T \left(b''_{yy}(t,Y(t)) D_s Y(t)\right)^2 ds \right)^{\frac p 2}\right] &\le \mathbb E\left[\xi^p \left( \int_0^T \left( D_s Y(t)\right)^2 ds \right)^{\frac p 2}\right]
        \\
        & \le \left(\mathbb E\left[\xi^{2p}\right]\right)^{\frac 1 2} \left(\mathbb E\left[\left( \int_0^T \left( D_s Y(t)\right)^2 ds \right)^{p}\right]\right)^{\frac 1 2}
        \\
        &< \infty.
    \end{align*}
    Therefore, by Lemma \ref{lemma: moments and Malliavin differentiability}, it is sufficient to prove that $b'_y (t, Y(t)) \in \mathbb D^{1,2}$ with \eqref{eq: Malliavin derivative of b'} being the corresponding Malliavin derivative.

    \textbf{Step 1.} Let $\phi \in C^{1}(\mathbb R)$ be a compactly supported function such that $\phi(x) = x$ whenever $|x|\le 1$ and $|\phi(x)| \le |x|$ for all $|x|>1$. Fix $t\in[0,T]$ and, for $m\ge 1$, put 
    \[
        f_m(y) := m\phi\left(\frac{b'_y(t,y)}{m}\right).
    \]
    Observe that
    \[
        f'_m(y) = b''_{yy}(t,y) \phi'\left(\frac{b'_y(t,y)}{m}\right)
    \]
    is bounded. Indeed, let $0 < \varepsilon_m < \psi(t) - \varphi(t)$ be such that
    \[
        -\frac{\gamma_1 \theta_1(t)}{\varepsilon_m^{\gamma_1+1}} + \max_{\varphi(t) \le x \le \psi(t)} a'_y(t,x) < m \inf\supp \phi
    \]
    and
    \[
        -\frac{\gamma_2 \theta_2(t)}{\varepsilon_m^{\gamma_2+1}} + \max_{\varphi(t) \le x \le \psi(t)} a'_y(t,x) < m \inf\supp \phi.
    \]
    Then, 
    \begin{itemize}
        \item if $y\in (\varphi(t), \varphi(t)+\varepsilon_m)$, then
        \begin{align*}
            b'_y(t,y) &= -\frac{\gamma_1 \theta_1(t)}{(y-\varphi(t))^{\gamma_1+1}} - \frac{\gamma_2\theta_2(t)}{(\psi(t) - y)^{\gamma_2+1}} + a'_y(t,y)
            \\
            & \le -\frac{\gamma_1 \theta_1(t)}{\varepsilon_m^{\gamma_1+1}} + \max_{\varphi(t) \le x \le \psi(t)}a'_y(t,x)
            \\
            &< m \inf\supp \phi,
        \end{align*}
        so $\frac{b'_y(t,y)}{m} \notin \supp\phi$, $f_m(y)=0$ and $f'_m(y) = 0$;

        \item if $y\in (\psi(t)-\varepsilon_m, \psi(t))$, then, similarly,
        \begin{align*}
            b'_y(t,y) &= -\frac{\gamma_1 \theta_1(t)}{(y-\varphi(t))^{\gamma_1+1}} - \frac{\gamma_2\theta_2(t)}{(\psi(t) - y)^{\gamma_2+1}} + a'_y(t,y)
            \\
            & \le -\frac{\gamma_2 \theta_2(t)}{\varepsilon_m^{\gamma_2+1}} + \max_{\varphi(t) \le x \le \psi(t)}a'_y(t,x)
            \\
            &< m \inf\supp \phi,
        \end{align*}
        so $\frac{b'_y(t,y)}{m} \notin \supp\phi$, $f_m(y) = 0$ and $f'_m(y) = 0$;

        \item on the compact set $[\varphi(t) +\varepsilon_m, \psi(t) - \varepsilon_m]$, both $f_m$ and its derivative $f'_m$ are continuous and hence bounded.
    \end{itemize}
    Therefore, the function $f_m$ satisfies the conditions of the classical Malliavin chain rule \cite[Proposition 1.2.3]{Nualart_2006}, so $f_m(Y(t)) \in \mathbb D^{1,2}$ and, with probability 1 for a.a. $s\in[0,T]$,
    \[
        D_s f_m(Y(t)) = b''_{yy}(t,Y(t)) \phi'\left(\frac{b'_y(t,Y(t))}{m}\right) D_sY(t).
    \]
    Now it remains to prove that $f_m(Y(t)) \to b'(t, Y(t))$ in $L^2(\Omega)$ and $D f_m(Y(t)) \to b''_{yy} (t, Y(t)) DY(t)$ in $L^2(\Omega\times[0,T])$ as $m\to \infty$ --- then the result will follow immediately from the closedness of the Malliavin derivative operator $D$.

    \textbf{Step 2:} $f_m(Y(t)) \to b'(t, Y(t))$ in $L^2(\Omega)$ as $m\to\infty$. By the definitions of $f_m$ and $\phi$, $f_m(Y(t)) \to b'(t, Y(t))$ a.s., $m\to\infty$. Moreover, with probability 1, $|f_m(Y(t))| \le |b'_y(t,Y(t))| \in L^2(\Omega)$ and hence the required convergence follows from the dominated convergence theorem.

    \textbf{Step 3:} $D f_m(Y(t)) \to b''_{yy} (t, Y(t)) DY(t)$ in $L^2(\Omega\times[0,T])$ as $m\to \infty$. By the definitions of $f_m$ and $\phi$, with probability 1,
    \begin{align*}
        \left(b''_{yy}(t,Y(t)) \phi'\left(\frac{b'_y(t,Y(t))}{m}\right)\right)^2 \int_0^T (D_sY(t))^2ds \to (b''_{yy}(t,Y(t)))^2 \int_0^T (D_sY(t))^2ds 
    \end{align*}
    as $m\to\infty$. Moreover, since $\phi$ has compact support, $\max_{y\in\mathbb R}(\phi'(y))^2 < \infty$, so we can write
    \begin{align*}
        \int_0^T (D_s f_m(Y(t)))^2ds &=\left(b''_{yy}(t,Y(t)) \phi'\left(\frac{b'_y(t,Y(t))}{m}\right)\right)^2 \int_0^T (D_sY(t))^2ds
        \\
        &\le \max_{y\in\mathbb R} (\phi'(y))^2 (b''_{yy}(t,Y(t)))^2 \int_0^T (D_sY(t))^2ds \in L^2(\Omega).
    \end{align*}
    Therefore, by the dominated convergence theorem,
    \begin{align*}
        \mathbb E&\left[ \int_0^T \left(D_s f_m(Y(t)) - b''_{yy} (t, Y(t)) D_sY(t)\right)^2ds \right] \to 0, \quad m \to \infty,
    \end{align*}
    which proves the first claim of the Proposition.

    2) Let us proceed with the second claim and verify that $\exp\left\{ \int_u^t b'_y(v,Y(v))dv \right\} \in \mathbb D^{1,p}$ with \eqref{eq: Malliavin derivative of exp(b')} being the corresponding Malliavin derivative. Note that, since $b'_y$ is bounded from above, $\exp\left\{ \int_u^t b'_y(v,Y(v))dv \right\}$ is also bounded from above and hence is an element of $L^p(\Omega)$ for any $p>1$. Moreover, by Proposition \ref{prop: bounds for derivatives}, boundedness of $\exp\left\{ \int_u^t b'_y(v,Y(v))dv \right\}$ and \eqref{proofeq: pathwise bound on the derivative}, we can write
    \begin{align*}
        \mathbb E&\left[\left(\int_0^T \left(\exp\left\{ \int_u^t b'_y(v,Y(v))dv \right\} \int_u^t b''_{yy}(v,Y(v)) D_s Y(v) dv\right)^2ds\right)^{\frac{p}{2}}\right]
        \\
        &\le C\mathbb E\left[\xi^p \left(\int_0^T  \int_u^t \left(D_s Y(v)\right)^2 dvds\right)^{\frac{p}{2}}\right]
        \\
        & \le C\mathbb E\left[\xi^p \left(\int_0^T  \int_u^t \mathcal K^2(v,s) dvds\right)^{\frac{p}{2}}\right] 
        + C\exp\left\{pcT\right\}\mathbb E\left[\xi^{2p}\right] \left(\int_0^T  \int_u^t   \int_s^v \mathcal K^2(u,s) du dvds\right)^{\frac{p}{2}}
        \\
        &<\infty
    \end{align*}
    and hence it is sufficient to prove that $\exp\left\{ \int_u^t b'_y(v,Y(v))dv \right\} \in \mathbb D^{1,2}$. 

    Since the Malliavin derivative operator $D$ is closed and the expression $\int_u^t b''_{yy}(v,Y(v)) D_s Y(v) dv$ is well-defined by Proposition \ref{prop: bounds for derivatives}, Step 1 of the current proof and Hille's theorem \cite[Theorem 1.2.4]{Veraar_Weis_2016} guarantee that $\int_u^t b'_{y}(v,Y(v)) Y(v) dv \in \mathbb D^{1,2}$ and
    \[
        D_s \int_u^t b'_{y}(v,Y(v)) Y(v) dv = \int_u^t b''_{yy}(v,Y(v)) D_s Y(v) dv.
    \]
    Finally, the function $x \mapsto e^x$ satisfies the conditions of the chain rule from \cite[Proposition 3.4]{DiNunno_Mishura_Yurchenko-Tytarenko_2022} and hence $\exp\left\{ \int_u^t b'_y(v,Y(v))dv \right\} \in \mathbb D^{1,2}$ and \eqref{eq: Malliavin derivative of exp(b')} holds.
\end{proof}

Proposition \ref{prop: auxiliary derivative computations} and Lemma \ref{lemma: product rule} together allow us to deduce the following corollary.

\begin{corollary}\label{cor: product rule for sandwiched processes}
    For any $0\le s < t \le T$ and $p>1$, $b'_y(s, Y(s)) \exp\left\{\int_s^t b'_y(v,Y(v))dv\right\} \in \mathbb D^{1,p}$ and
    \begin{equation}
    \begin{aligned}
        D_u&\left[b'_y(s, Y(s)) \exp\left\{\int_s^t b'_y(v,Y(v))dv\right\}\right] 
        \\
        &= b''_{yy}(s,Y(s))\exp\left\{\int_s^t b'_y(v,Y(v))dv\right\} D_uY(s) 
        \\
        &\quad + b'_y(s, Y(s)) \exp\left\{ \int_s^t b'_y(v,Y(v))dv \right\} \int_s^t b''_{yy}(v,Y(v)) D_u Y(v) dv.
    \end{aligned}
    \end{equation}
\end{corollary}
\begin{proof}
    For fixed $0\le s < t \le T$, denote
    \[
        X_1 := b'_y(s, Y(s)), \quad X_2 := \exp\left\{\int_s^t b'_y(v,Y(v))dv\right\}.
    \]
    By Proposition \ref{prop: auxiliary derivative computations} and Lemma \ref{lemma: product rule} from the Appendix, it is sufficient to check that for all $p\ge 2$
    \begin{itemize}
        \item[(i)] the product $X_1 X_2 \in L^p(\Omega)$,
        \item[(ii)] $\mathbb E\left[ \left( \int_0^T (X_2 D_u X_1)^2 \right)^{\frac p 2} \right] < \infty$ and 
        \item[(iii)] $\mathbb E\left[ \left( \int_0^T (X_1 D_u X_2)^2 \right)^{\frac p 2} \right] < \infty$.
    \end{itemize}
    All conditions (i)--(iii) can be checked in a straightforward manner using Proposition \ref{prop: bounds for derivatives} and the arguments similar to the proof of Proposition \ref{prop: D1p}.
\end{proof}

We are now ready to formulate the main result of this section.

\begin{theorem}
    For any $t\in[0,T]$ and $p\ge 2$,
    \begin{itemize}
        \item[1)] $Y(t) \in \mathbb D^{2,p}$,

        \item[2)] with probability 1 and for a.a. $r,s\in[0,T]$,
        \begin{equation}\label{eq: second Malliavin derivative}
        \begin{aligned}
            D_r D_s Y(t) &= \int_s^t  \mathcal K(u,s) F_1(t,u)  \left(\int_u^t b''_{yy}(v, Y(v)) D_r Y(v)dv\right)du 
            \\
            & \qquad + \int_s^t \mathcal K(u,s) F_2(t,u) D_r Y(u) du,
        \end{aligned}
        \end{equation}
        where 
        \begin{equation*}
        \begin{aligned}
            F_1(t,u) &:= b'_{y}(u, Y(u)) \exp\left\{ \int_u^t b'_y(v, Y(v))dv \right\}, 
            \\
            F_2(t,u) &:= b''_{yy}(u, Y(u)) \exp\left\{ \int_u^t b'_y(v, Y(v))dv\right\}.
        \end{aligned}    
        \end{equation*}
    \end{itemize}
\end{theorem}

\begin{proof}
    Our goal is to prove that $Y(t) \in \mathbb D^{2,p}$ and
    \begin{align*}
        D_rD_sY(t) &= \int_s^t \mathcal K(u,s) D_r\left[ b'_y(u, Y(u)) \exp\left\{\int_u^t b'_y(v, Y(v))dv\right\} \right]du
        \\
        & = \int_s^t \mathcal K(u,s) D_r\left[ F_1(t,u) \right]du
    \end{align*}
    since, in such case, \eqref{eq: second Malliavin derivative} follows immediately from Corollary \ref{cor: product rule for sandwiched processes}. Recall that
    \begin{equation*}
    \begin{aligned}
        D_sY(t) = \mathcal K(t,s) + \int_s^t \mathcal K(u,s) F_1(t,u)du.
    \end{aligned}    
    \end{equation*}
    Clearly, for any $0\le r , s < t\le T$, 
    \[
        D_r \mathcal K(t,s) = 0,
    \]
    so, by closedness of $D$ and Hille's theorem \cite[Theorem 1.2.4]{Veraar_Weis_2016}, it is enough to show that
    \begin{itemize}
        \item[(i)] for a.a. $0\le s \le u < t\le T$, $\mathcal K(u,s) F_1(t,u) \in \mathbb D^{1,p}$ and

        \item[(ii)] for a.a. $0\le s < t\le T$, 
        \begin{equation*}
        \begin{aligned}
            \int_0^T &\left(\mathbb E\left[\left(\int_0^T (D_r[\mathcal K(u,s) F_1(t,u)])^2 dr\right)^{\frac{p}{2}}\right]\right)^{\frac{1}{p}}du 
            \\
            &= \int_0^T \mathcal K(u,s) \left(\mathbb E\left[\left(\int_0^T (D_r[ F_1(t,u)])^2 dr\right)^{\frac{p}{2}}\right]\right)^{\frac{1}{p}}du
            \\
            &< \infty.
        \end{aligned}
        \end{equation*}
    \end{itemize}
    Item (i) above follows immediately from Corollary \ref{cor: product rule for sandwiched processes}. As for item (ii), observe that, by Proposition \ref{prop: bounds for derivatives}, \eqref{proofeq: pathwise bound on the derivative} as well as the boundedness of $\exp\left\{\int_u^t b'_y(v,Y(v))dv\right\}$, we have
    \begin{align*}
        (D_r[ F_1(t,u) ])^2 &\le C\bigg( (b''_{yy}(u,Y(u)))^2\exp\left\{2\int_u^t b'_y(v,Y(v))dv\right\} (D_rY(u))^2 
        \\
        &\quad + (b'_y(u, Y(u)))^2 \exp\left\{ 2\int_u^t b'_y(v,Y(v))dv \right\} \int_u^t (b''_{yy}(v,Y(v)) D_r Y(v))^2 dv\bigg)
        \\
        & \le C \left(\xi^2 (D_rY(u))^2 + \xi^4 \int_u^t (D_rY(v))^2 dv \right)
        \\
        &\le C\xi^2 \left(\mathcal K^2(u,r) + \int_r^u \mathcal K^2(z,r) dz\right) 
        \\
        &\quad + C  \xi^4 \left(\int_u^t \mathcal K^2(v,r)dv + \int_u^t\int_r^v \mathcal K^2(z,r) dzdv\right).
    \end{align*}
    Hence, for any $p\ge 2$, Remark \ref{rem: K} implies
    \begin{align*}
        \int_0^T (D_r[ F_1(t,u)])^2 dr &\le C\xi^2 \left(\int_0^T\mathcal K^2(u,r)dr + \int_0^T\int_r^u \mathcal K^2(z,r) dzdr\right) 
        \\
        &\quad + C  \xi^4 \left(\int_0^T\int_u^t \mathcal K^2(v,r)dvdr + \int_0^T\int_u^t\int_r^v \mathcal K^2(z,r) dzdvdr\right)
        \\
        & \le C\left(\xi^2+\xi^4\right),
    \end{align*}
    so, since $\xi$ has moments of all orders, (ii) holds, which finalizes the proof.
\end{proof}

Finally, denote $\mathbb L^{2,p} := L^p([0,T];\mathbb D^{2,p})$. We complete the section with the following result.
\begin{corollary}
    For any $p\ge 2$, $Y \in \mathbb L^{2,p}$.
\end{corollary}
\begin{proof}
    By the definition of the $\|\cdot\|_{2,p}$-norm in \eqref{app: Dkp norm} from Appendix \ref{appendix}, it is sufficient to check that
    \begin{equation}\label{proofeq: L1}
        \int_0^T \mathbb E[|Y(t)|^p] < \infty,
    \end{equation}
    \begin{equation}\label{proofeq: L2}
        \int_0^T \mathbb E\left[ \left(\int_0^T (D_s Y(t))^2ds\right)^{\frac{p}{2}} \right]dt < \infty
    \end{equation}
    and
    \begin{equation}\label{proofeq: L3}
        \int_0^T \mathbb E\left[ \left(\int_0^T \int_0^T(D_rD_s Y(t))^2dsdr\right)^{\frac{p}{2}} \right]dt < \infty.
    \end{equation}
    By \eqref{eq: sandwich property}, \eqref{proofeq: L1} holds automatically. Next, \eqref{proofeq: L2} can be easily deduced from \eqref{proofeq: bound for norm of D1}. Finally, using Proposition \eqref{prop: bounds for derivatives} and the boundedness of $\exp\left\{\int_u^t b'_y(v, Y(v))dv\right\}$, it is easy to prove a bound similar to \eqref{proofeq: bound for norm of D1} for
    \[
        \mathbb E\left[ \left(\int_0^T \int_0^T(D_rD_s Y(t))^2dsdr\right)^{\frac{p}{2}} \right]
    \]
    which implies \eqref{proofeq: L3}. By this, the proof is complete.
\end{proof}

\section{Power law in VSV model}\label{sec: power law}

Having the second-order Malliavin differentiability in place, we now possess all the necessary tools to analyze the behavior of implied volatility skew of a model with the sandwiched process \eqref{eq: sandwiched SDE} as stochastic volatility. Namely, we consider a (risk-free) market model with the price process $S = \{S(t),~t\in[0,T]\}$ of the form
\begin{equation}\label{eq: model}
\begin{aligned}
    S(t) &= e^{X(t)},
    \\
    X(t) &= x_0 + rt -\frac{1}{2}\int_0^t Y^2(s)ds + \int_0^t Y(s) \left(\rho dB_1(s) + \sqrt{1-\rho^2}dB_2(s)\right),
    \\
    Y(t) & = y_0 + \int_0^t b(s, Y(s))ds + \int_0^t \mathcal K(t,s)dB_1(s),
\end{aligned}
\end{equation}
where $B_1$, $B_2$ are two independent Brownian motions, $X = \{X(t),~t\in[0,T]\}$ denotes the (risk-free) log-price of some asset starting from some level $x_0\in\mathbb R$, $r$ is a constant instantaneous interest rate, and $\rho\in(-1,1)$ is a correlation coefficient that accounts for the leverage effect. As previously, the drift $b$ and the Volterra kernel $\mathcal K$ satisfy Assumptions \ref{assum: K} and \ref{assum: b}.

\begin{remark}
    The model \eqref{eq: model} was initially introduced in \cite{DiNunno_Mishura_Yurchenko-Tytarenko_2022} and, given the nature of the volatility process, is called the \textit{Sandwiched Volterra Volatility} (SVV) model.
\end{remark}

To establish the conditions under which \eqref{eq: model} gives power law of the short-term at-the-money implied volatility, we will apply the fundamental result \cite[Theorem 6.3]{Alos_Leon_Vives_2007} which connects the shape of the skew with the Malliavin derivative of the volatility. 

\begin{remark}\label{rem: skew notation}
    In the recent literature (see e.g. \cite{Bayer_Friz_Gatheral_2016, Delemotte_Marco_Segonne_2023, Di_Nunno_Kubilius_Mishura_Yurchenko-Tytarenko_2023, Gatheral_Jaisson_Rosenbaum_2014}), it is typical to characterize the implied volatility skew in terms of $\frac{\partial \widehat \sigma}{\partial \kappa}$ with $\kappa = \log\frac{K}{e^{r\tau + x_0}}$ being the log-moneyness. In \cite{Alos_Leon_Vives_2007}, a slightly different parametrization $\widehat\sigma_{\text{log-price}}(\tau, x_0)$ is considered with
    \[
        \widehat\sigma_{\text{log-price}}(\tau, x) = \widehat\sigma\left(\tau, \log\frac{K}{e^{r\tau}}-x\right).
    \]
    With this parametrization,
    \[
        \frac{\partial\widehat\sigma_{\text{log-price}}(\tau, x)}{\partial x} = -\frac{\partial \widehat\sigma\left(\tau, \log\frac{K}{e^{r\tau}}-x\right)}{\partial \kappa},
    \]
    i.e. the typically negative at-the-money skews for $\widehat\sigma$ are equivalent to positive $ \frac{\partial\widehat\sigma_{\text{log-price}}(\tau, x)}{\partial x}$ at $x = \log\frac{K}{e^{r\tau}}$ and the power law \eqref{eq: theoretical power law intro} is equivalent to
    \begin{equation*}
        \left|\frac{\partial \widehat \sigma_{\text{log-price}}}{\partial x} (\tau, x)\right|_{x = \log\frac{K}{e^{r\tau}}} = O(\tau^{-\frac{1}{2} + H}), \quad \tau \to 0.
    \end{equation*}
\end{remark}

With Remark \ref{rem: skew notation} in mind, let us provide a slightly adjusted version of \cite[Theorem 6.3]{Alos_Leon_Vives_2007}.

\begin{theorem}\label{th: power law and derivative}
    Consider a risk-free log-price
    \begin{equation}\label{eq: general model}
        X(t) = x_0 + rt -\frac{1}{2}\int_0^t \sigma^2(s)ds + \int_0^t \sigma(s) \left(\rho dB_1(s) + \sqrt{1-\rho^2}dB_2(s)\right),
    \end{equation}
    where $B_1$, $B_2$ are two independent Brownian motions, $x_0\in\mathbb R$ is a deterministic initial value, $r$ is an instantaneous interest rate, $\rho\in(-1,1)$ is a correlation coefficient and $\sigma = \{\sigma(t),~t\in[0,T]\}$ is a square-integrable stochastic process with right-continuous trajectories adapted to the filtration $\mathcal F = \{\mathcal F_t,~t\in[0,T]\}$ generated by $B_1$. 
    
    Assume that
    \begin{itemize}
        \item[(H1)] $\sigma \in \mathbb L^{2,4}$ with respect to $B_1$;
        \item[(H2)] there exists a constant $\varphi_*>0$ such that, with probability 1, $\sigma(t)>\varphi_*$ for all $t\in[0,T]$;
        \item[(H3)] there exists a constant $H\in\left(0, \frac 1 2\right)$ such that, with probability 1, for any $0 < s < t < T$,
        \begin{align}
            \mathbb E\left[ (D_s \sigma(t))^2\right] &\le \frac{C}{(t-s)^{1-2H}},\label{eq: D1 Holder}
            \\
            \mathbb E\left[(D_r D_s \sigma(t))^2\right] & \le C\left(\frac{t-r}{t-s}\right)^{1 - 2H},\label{eq: D2 Holder}
        \end{align}    
        where $C>0$ is some constant;
        
        \item[(H4)] $\sigma$ has a.s. right-continuous trajectories;
        \item[(H5)] $\sup_{r,s,t \in[0,\tau]} \mathbb E\left[(\sigma(s)\sigma(t)-\sigma^2(r))^2\right] \to 0$ when $\tau \to 0+$.
    \end{itemize}
    Finally, assume that there exists a constant $K_{\sigma} > 0$ such that, with probability 1,
    \begin{equation}\label{eq: Mal der limit}
        \frac{1}{\tau^{\frac{3}{2} + H}} \int_{0}^\tau \int_s^\tau \mathbb E\left[ D_s \sigma(t)\right]dtds - K_{\sigma} \to 0, \quad \tau \to 0+.
    \end{equation}
    Then, with probability 1,
    \[
        \lim_{\tau \to 0} \tau^{\frac{1}{2} - H} \frac{\partial \widehat \sigma_{\text{log-price}}}{\partial x} (\tau, x)  \bigg|_{x = \log\frac{K}{e^{r\tau}}}  = - \frac{\rho}{\sigma(0)} K_{\sigma}.
    \]
    In particular, if $\rho K_{\sigma} < 0$, the at-the-money implied volatility skew exhibits the power law behavior with the correct sign of the skew.
\end{theorem}

\begin{remark}
    The original formulation of  \cite[Theorem 6.3]{Alos_Leon_Vives_2007} is slightly more general than Theorem \ref{th: power law and derivative} above in the sense that
    \begin{itemize}
        \item[1)] in \cite[Theorem 6.3]{Alos_Leon_Vives_2007}, the log-price $X$ is allowed to have jumps; 

        \item[2)] the result in \cite{Alos_Leon_Vives_2007} is formulated for the \textit{future} implied volatility surfaces $\widehat \sigma_{\text{log-price}}(t_0, \tau, X(t_0))$, $t_0\ge 0$.
    \end{itemize}
    Since we are interested in the continuous model \eqref{eq: model}, we removed the jump component in \eqref{eq: general model} and, for the simplicity of notation, we put $t_0=0$.
\end{remark}

Observe that the SVV model \eqref{eq: model} automatically satisfies a number of assumptions of Theorem \ref{th: power law and derivative}:
\begin{itemize}
    \item assumption (H2) with $\varphi^* := \min_{t\in[0,T]} \varphi(t) > 0$;
    \item assumption (H4) since $Y$ is continuous a.s.;
    \item assumption (H1) by the results of Section \ref{sec: Malliavin} above.
\end{itemize}
Therefore, it remains to check (H3), (H5), and \eqref{eq: Mal der limit}. Naturally, given the shape of the Malliavin derivative \eqref{eq: Malliavin derivative of Y(t)}, both (H3) and \eqref{eq: Mal der limit} require additional assumptions on the kernel, so let us start with (H5).

\begin{proposition}
    Let Assumptions \ref{assum: K} and \ref{assum: b} hold. Then with probability 1,
    \begin{align*}
        \sup_{r,s,t \in[0,\tau]} \mathbb E\left[(Y(s)Y(t)-Y^2(r))^2\right] \to 0, \quad \tau \to 0.
    \end{align*}
\end{proposition}
\begin{proof}
    By \cite[Lemma 3.6]{DiNunno_Mishura_Yurchenko-Tytarenko_2023b}, there exists a positive random variable $\Upsilon = \Upsilon_T$ such that for all $t_1,t_2\in[0,T]$
    \[
        |Y(t_1) - Y(t_2)| \le \Upsilon |t_1-t_2|^{\lambda}
    \]
    and, for any $r>0$,
    \[
        \mathbb E[\Upsilon^r] < \infty.
    \]
    Therefore, given that $\max_{t\in[0,T]}Y(t) < \max_{t\in[0,T]}\psi(t)$ by \eqref{eq: sandwich property},
    \begin{align*}
        \mathbb E&\left[(Y(s)Y(t)-Y^2(r))^2\right] 
        \\
        &= \mathbb E\left[(Y(s)(Y(t)-Y(r)) + Y(r)(Y(s) - Y(r)))^2\right]
        \\
        & \le 2 \mathbb E\left[(Y^2(s)(Y(t)-Y(r))^2\right] + 2 \mathbb E\left[Y^2(r)(Y(s) - Y(r))^2\right]
        \\
        &\le 2|t-r|^{2\lambda} \max_{s\in[0,T]} \psi^2(s)\mathbb E\left[\Upsilon^2\right] + 2|s-r|^{2\lambda} \max_{s\in[0,T]} \psi^2(s)\mathbb E\left[\Upsilon^2\right]
    \end{align*}
    and hence, with probability 1,
    \begin{align*}
        \sup_{r,s,t \in[0,\tau]} \mathbb E\left[(Y(s)Y(t)-Y^2(r))^2\right] & \le 4 \tau^{2\lambda} \max_{s\in[0,T]} \psi^2(s)\mathbb E\left[\Upsilon^2\right] \to 0
    \end{align*}
    as $\tau \to 0+$.
\end{proof}

Our next step is to handle \eqref{eq: Mal der limit}.
\begin{proposition}
    Let Assumptions \ref{assum: K} and \ref{assum: b} hold and the Volterra kernel $\mathcal K$ be such that
    \begin{equation}\label{eq: Kernel explosion}
        \frac{1}{\tau^{\frac{3}{2} + H}} \int_{0}^\tau \int_s^\tau \mathcal K(t,s) dtds \to K_Y, \quad \tau \to 0+,
    \end{equation}
    where $K_Y$ is some finite constant. Then, with probability 1
    \[
        \frac{1}{\tau^{\frac{3}{2} + H}} \int_{0}^\tau \int_s^\tau \mathbb E\left[ D_s Y(t)\right]dtds - K_Y \to 0, \quad \tau \to 0+.
    \]
    
\end{proposition}
\begin{proof}
    Recall that
    \[
        F_1(t,u) := b'_{y}(u, Y(u)) \exp\left\{ \int_u^t b'_y(v, Y(v))dv\right\}
    \]
    and that, by Proposition \eqref{prop: bounds for derivatives},
    \begin{equation}\label{eq: F1 bound}
        |F_1(t,u)| \le e^{cT}\xi, 
    \end{equation}
    where $c:= \max_{(t,y)\in\overline{\mathcal D}}a'_y(t,y)$. Then we can write
    \begin{align*}
        \frac{1}{\tau^{\frac{3}{2} + H}}& \int_{0}^\tau \int_s^\tau \mathbb E\left[ D_s Y(t)\right]dtds 
        \\
        & =  \frac{1}{\tau^{\frac{3}{2} + H}} \int_{0}^\tau \int_s^\tau \mathcal K(t,s) dtds 
        \\
        &\quad + \frac{1}{\tau^{\frac{3}{2} + H}}\int_{0}^\tau \int_s^\tau \int_s^t \mathcal K(u,s) \mathbb E\left[ F_1(t,u)\right]dudtds
        \\
        & = \frac{1}{\tau^{\frac{3}{2} + H}} \int_{0}^\tau \int_s^\tau \mathcal K(t,s) dtds 
        \\
        &\quad + \frac{1}{\tau^{\frac{3}{2} + H}}\int_{0}^\tau \int_s^\tau \mathcal K(u,s) \left(\int_u^\tau  \mathbb E\left[ F_1(t,u)\right]dt\right) du ds.
    \end{align*}
    The term $\frac{1}{\tau^{\frac{3}{2} + H}} \int_{0}^\tau \int_s^\tau \mathcal K(t,s) dtds$ converges to $K_Y$ by \eqref{eq: Kernel explosion}. As for the second term, note that, with probability 1, for any $u\in[0, \tau]$, 
    \begin{align*}
        \int_u^\tau |\mathbb E\left[ F_1(t,u)\right]|dt &\le C\mathbb E\left[\xi\right] \tau  
    \end{align*}
    and hence, given \eqref{eq: Kernel explosion}, with probability 1,
    \[
        \frac{1}{\tau^{\frac{3}{2} + H}}\int_{0}^\tau \int_s^\tau \mathcal K(u,s) \left(\int_u^\tau  \mathbb E\left[ F_1(t,u)\right]dt\right) du ds \to 0, \quad \tau\to 0+,
    \]
    which ends the proof.
\end{proof}

Finally, let us deal with (H3). 
\begin{proposition}
    Let Assumptions \ref{assum: K} and \ref{assum: b} hold with $H\in \left(\frac 1 6, \frac{1}{2}\right)$ and the Volterra kernel $\mathcal K$ be such that for any $0\le s < t \le T$
    \begin{equation}\label{eq: Kernel Holder}
        |\mathcal K(t,s)| \le C|t-s|^{-\frac{1}{2}+H}
    \end{equation}
    for some constant $C>0$. Then the hypothesis (H3) from Theorem \ref{th: power law and derivative} holds for the volatility process $\sigma = Y$.
\end{proposition}
\begin{proof}
    Fix $0< r,s < t$. Then, taking into account \eqref{eq: F1 bound}, with probability 1,
    \begin{equation}\label{eq: DY Holder}
    \begin{aligned}
        |D_s Y(t)| & \le |\mathcal K(t,s)| + \int_s^t |\mathcal K(u,s)| |F_1(t,u)| du
        \\
        &\le C\left(|t-s|^{-\frac{1}{2}+H} + \xi \int_s^t |u-s|^{-\frac{1}{2}+H}du \right)
        \\
        &\le C(1+T\xi)|t-s|^{-\frac{1}{2}+H}
        \\
        &=: \zeta |t-s|^{-\frac{1}{2}+H},
    \end{aligned}    
    \end{equation}
    which immediately implies \eqref{eq: D1 Holder}. Next, by Proposition \ref{prop: bounds for derivatives},
    \[
        |b''_{yy}(v, Y(v))| \le \xi
    \]
    for any $v\in[0,T]$ and, for any $0\le u \le t \le T$,
    \begin{align*}
        |F_2(t,u)| &= \left|b''_{yy}(u, Y(u)) \exp\left\{ \int_u^t b'_y(v, Y(v))dv\right\}\right|
        \\
        & \le e^{cT}\xi 
    \end{align*}
    with $c:= \max_{(t,y)\in\overline{\mathcal D}}a'_y(t,y)$, so we can write
    \begin{align*}
        |D_r D_s Y(t)| & \le \int_s^t  |\mathcal K(u,s)| |F_1(t,u)|  \left(\int_u^t |b''_{yy}(v, Y(v))| |D_r Y(v)|dv\right)du 
        \\
        & \quad + \int_s^t |\mathcal K(u,s)| |F_2(t,u)| |D_r Y(u)| du
        \\
        &\le C\left(\xi^2  \int_s^t |\mathcal K(u,s)| \left(\int_u^t |D_r Y(v)|dv\right)du + \xi \int_s^t |\mathcal K(u,s)| |D_r Y(u)| du\right)
        \\
        & = C\left(\xi^2  \int_s^t |\mathcal K(u,s)| \left(\int_{u\vee r}^t |D_r Y(v)|dv\right)du + \xi \int_{r \vee s}^t |\mathcal K(u,s)| |D_r Y(u)| du\right).
    \end{align*}
    Taking into account \eqref{eq: Kernel Holder} and \eqref{eq: DY Holder},
    \begin{align*}
        |D_r D_s Y(t)| & \le C\bigg(\xi^2\zeta  \int_s^t |u-s|^{-\frac 1 2 + H} \left(\int_{u \vee r}^t |v-r|^{-\frac 1 2 + H} dv\right)du 
        \\
        &\qquad + \xi \zeta \int_{r \vee s}^t |u-s|^{-\frac 1 2 + H} |u-r|^{-\frac 1 2 + H} du\bigg)
        \\
        &\le  C\bigg(\xi^2\zeta  \int_s^t |u-s|^{-\frac 1 2 + H} |t-r|^{\frac 1 2 + H} du
        \\
        &\qquad + \xi \zeta \int_{r \vee s}^t |u-s|^{-\frac 1 2 + H} |u-r|^{-\frac 1 2 + H} du\bigg).
    \end{align*}
    Note that
    \begin{align*}
        \int_s^t |u-s|^{-\frac 1 2 + H} |t-r|^{\frac 1 2 + H} du & \le C |t-r|^{\frac 1 2 + H} |t-s|^{\frac 1 2 + H}
        \\
        &\le C \left(\frac{t-r}{t-s}\right)^{\frac 1 2 - H}.
    \end{align*}
    As for the integral $\int_{r \vee s}^t |u-s|^{-\frac 1 2 + H} |u-r|^{-\frac 1 2 + H} du$, we have two cases:
    \begin{itemize}
        \item if $0 < r \le s < t$, we can write
        \begin{align*}
            \int_s^t |u-s|^{-\frac 1 2 + H} |u-r|^{-\frac 1 2 + H} du & \le \int_s^t |u-s|^{-1 + 2H}du
            \\
            & \le C(t-s)^{2H} \le C \left(\frac{t-r}{t-s}\right)^{\frac 1 2 - H};
        \end{align*}

        \item similarly, if $0 < s < r < t$ and given that $H>\frac 1 6$, we have
        \begin{align*}
            \int_r^t |u-s|^{-\frac 1 2 + H} |u-r|^{-\frac 1 2 + H} du & \le \int_r^t |u-r|^{-1 + 2H}du
            \\
            & \le C(t-r)^{2H} \le C \left(\frac{t-r}{t-s}\right)^{\frac 1 2 - H}.
        \end{align*}
        
    \end{itemize}
    In any case,
    \begin{align*}
        |D_r D_s Y(t)| & \le  C\xi \zeta (\xi + 1) \left(\frac{t-r}{t-s}\right)^{\frac 1 2 - H},
    \end{align*}
    where $\xi$ and $\zeta$ are random variables having all moments, and hence \eqref{eq: D2 Holder} holds.
\end{proof}

All the findings of this Section can now be summarized in the following theorem which should be regarded as the main result of the paper.

\begin{theorem}\label{th: main theorem}
     Let Assumptions \ref{assum: K} and \ref{assum: b} hold with $H\in \left(\frac 1 6, \frac{1}{2}\right)$. Assume also that $\rho \ne 0$ in \eqref{eq: model} and that the Volterra kernel $\mathcal K$ is such that for any $0 \le s < t \le T$
    \begin{equation*}
        |\mathcal K(t,s)| \le C|t-s|^{-\frac{1}{2}+H}
    \end{equation*}
    for some constant $C>0$ and \eqref{eq: Kernel explosion} holds with $\rho K_Y < 0$. Then the SVV model \eqref{eq: model} reproduces the power law of the at-the-money implied volatility skew with the correct sign.
\end{theorem}

\begin{example}
    Let $\frac 1 6 < H_0 < H_1 < ... < H_n <1$ be such that $H_0 < \frac{1}{2}$ and $\alpha_k > 0$, $k=0,...,n$. Then the kernel
    \[
        \mathcal K(t,s) = \left(\sum_{k=0}^n \alpha_k(t-s)^{H_k - \frac{1}{2}}\right)\mathbbm 1_{s<t}
    \]
    satisfies the assumptions of Theorem \ref{th: main theorem}, so the corresponding SVV model generates power law \eqref{eq: theoretical power law intro} with $H= H_0$ provided that $\rho<0$ in \eqref{eq: model}.
\end{example}

\appendix

\section{Selected results from Malliavin calculus}\label{appendix}

\subsection{Malliavin derivative and the space $\mathbb D^{k,p}$}

Hereafter, we summarize the essentials of the Malliavin derivative with respect to the classical Brownian motion. For more details, we refer the reader to the classical books \cite{Nualart_2006} or \cite{Nunno_Øksendal_Proske_2009}.

Denote $C^{(\infty)}_p(\mathbb R^n)$ the space of all infinitely differentiable functions with the derivatives of at most polynomial growth. Let $B = \{B(t),~t\in[0,T]\}$ be a standard Brownian motion. For any $h\in L^2([0,T])$, denote
\[
    B(h) := \int_0^T h(t) d B(t).
\]

\begin{definition}
    The random variables $X$ of the form
    \[
        X = f(B(h_1), ..., B(h_n)),
    \]
    where $n\ge 1$, $f \in C^{(\infty)}_p(\mathbb R^n)$ and $h_1,..., h_n \in L^2([0,T])$ are called smooth. The set of all smooth random variables is denoted by $\mathcal S$.
\end{definition}

\begin{definition}
    Let $X\in \mathcal S$. The Malliavin derivative of $X$ (with respect to $B$) is the $L^2([0,T])$-valued random variable of the form
    \[
        DX := \sum_{k=1}^n \frac{\partial f}{\partial x_k}(B(h_1), ..., B(h_n)) h_k.
    \]
\end{definition}

By \cite[Proposition 1.2.1]{Nualart_2006}, the operator $D$ is closable from $L^p(\Omega)$ to $L^p(\Omega\times[0,T])$ for any $p\ge 1$, and we use the same notation $D$ for the closure. The domain of this closure $D$ in $L^p(\Omega)$, i.e. the closure of the class $\mathcal S$ with respect to the norm 
\[
    \| X \|_{1,p} := \left( \mathbb E\left[|X|^p\right] + \mathbb E\left[ \left(\int_0^T (D_s X)^2ds\right)^{\frac{p}{2}} \right] \right)^{\frac{1}{p}},
\]
is traditionally denoted by $\mathbb D^{1,p}$. This definition can be iterated as described in \cite[p. 27]{Nualart_2006} to introduce the iterated derivative $D^k X$ as a random variable with values in $(L^2([0,T]))^{\otimes k} \sim L^2([0,T]^k)$. One can also define $\mathbb D^{k,p}$ as the completion of $\mathcal S$ with respect to the seminorm
\begin{equation}\label{app: Dkp norm}
    \| X \|_{k,p} := \left( \mathbb E\left[|X|^p\right] + \sum_{j=1}^k \mathbb E\left[ \|D^j X\|_{L^2([0,T]^k)}^p \right] \right)^{\frac{1}{p}}.
\end{equation}

Throughout the paper, we often use the following lemma which is essentially a simplified version of \cite[Proposition 1.5.5]{Nualart_2006}.

\begin{lemma}\label{lemma: moments and Malliavin differentiability}
    Let $p>1$ and $X\in\mathbb D^{1,2}$ be such that 
    \[
        \mathbb E\left[|X|^p\right] < \infty
    \]
    and
    \[
        \mathbb E\left[\left( \int_0^T (D_s X)^2 ds \right)^{\frac p 2}\right] < \infty.
    \]
    Then $X\in \mathbb D^{1,p}$.
\end{lemma}

\subsection{Generalized Malliavin product rule}

Finally, let us prove a generalized version of the product rule from \cite[Exercise 1.2.12]{Nualart_2006} or \cite[Theorem 3.4]{Nunno_Øksendal_Proske_2009}.

\begin{lemma}\label{lemma: product rule}
    Let $X_1$, $X_2 \in \mathbb D^{1,2}$ be such that
    \begin{itemize}
        \item[(i)] $X_1X_2 \in L^2(\Omega)$;
        \item[(ii)] $X_2DX_1$, $X_1DX_2 \in L^2(\Omega\times[0,T])$.
    \end{itemize}
    Then $X_1X_2\in \mathbb D^{1,2}$ and
    \[
        D[X_1X_2] = X_2DX_1 + X_1DX_2.
    \]
    If, in addition, 
    \begin{equation*}
        \mathbb E[|X_1X_2|^p] < \infty, \quad \mathbb E\left[ \left( \int_0^T ( X_2D_uX_1 + X_1D_uX_2 )^2du \right)^{\frac p 2} \right] < \infty
    \end{equation*}
    for some $p\ge 2$, then $X_1X_2 \in \mathbb D^{1,p}$. 
\end{lemma}
\begin{proof}
    Let $\phi\in C^\infty(\mathbb R)$ be a compactly supported function such that $\phi(x) = x$ whenever $|x|\le 1$ and $|\phi(x)| \le |x|$ for all $|x|>1$. For $m\ge 1$, put
    \[
        f_m(x_1, x_2) := m^2 \phi\left(\frac{x_1}{m}\right)\phi\left(\frac{x_2}{m}\right)
    \]
    and observe that both partial derivatives
    \[
        \frac{\partial f_m}{\partial x_1} (x_1,x_2) = m\phi'\left(\frac{x_1}{m}\right)\phi\left(\frac{x_2}{m}\right), \quad \frac{\partial f_m}{\partial x_2} (x_1,x_2) = m\phi\left(\frac{x_1}{m}\right)\phi'\left(\frac{x_2}{m}\right)
    \]
    are bounded. Therefore, by the classical chain rule \cite[Proposition 1.2.3]{Nualart_2006}, 
    \[
        Df_m(X_1, X_2) = m\left(\phi'\left(\frac{X_1}{m}\right)\phi\left(\frac{X_2}{m}\right) DX_1 + \phi\left(\frac{X_1}{m}\right)\phi'\left(\frac{X_2}{m}\right) DX_2\right).
    \]
    Now it is sufficient to prove that 
    \begin{equation}\label{proofeq: fm to x1x2}
        f_m(X_1,X_2) \to X_1X_2
    \end{equation}
    in $L^2(\Omega)$ and 
    \begin{equation}\label{proofeq: fm' to Dx1x2}
        m\left(\phi'\left(\frac{X_1}{m}\right)\phi\left(\frac{X_2}{m}\right) DX_1 + \phi\left(\frac{X_1}{m}\right)\phi'\left(\frac{X_2}{m}\right) DX_2\right) \to X_2DX_1 + X_1DX_2
    \end{equation}
    in $L^2(\Omega\times[0,T])$ as $m\to\infty$.

    Observe that $|f_m(X_1,X_2)| \to X_1X_2$ a.s. as $m\to\infty$ and
    \[
        |f_m(X_1,X_2)| \le X_1 X_2 \in L^2(\Omega),
    \]
    so \eqref{proofeq: fm to x1x2} holds by the dominated convergence theorem. Next, since $\phi'$ is bounded, we have that, with probability 1,
    \begin{equation*}
    \begin{aligned}
        m&\left|\phi'\left(\frac{X_1}{m}\right)\phi\left(\frac{X_2}{m}\right) DX_1 + \phi\left(\frac{X_1}{m}\right)\phi'\left(\frac{X_2}{m}\right) DX_2\right|
        \\
        &\le \max_{x\in\mathbb R}|\phi'(x)| \left( |X_2 DX_1| + |X_1DX_2|  \right) \in L^2(\Omega\times[0,T]).
    \end{aligned}    
    \end{equation*}
    Therefore, since $m\phi'\left(\frac{X_1}{m}\right)\phi\left(\frac{X_2}{m}\right) \to X_2$ a.s. and $m\phi\left(\frac{X_1}{m}\right)\phi'\left(\frac{X_2}{m}\right) \to X_1$ a.s. as $m\to\infty$, \eqref{proofeq: fm' to Dx1x2} holds by the dominated convergence, which ends the proof of the first claim.
    
    The second claim immediately follows from Lemma \ref{lemma: moments and Malliavin differentiability}.
\end{proof}

\bibliographystyle{acm}
\bibliography{biblio.bib}

\end{document}